\documentclass[12pt]{article}

\usepackage[T1]{fontenc}
\usepackage[latin1]{inputenc}
\usepackage[english]{babel}
\usepackage{amsmath,amssymb,amsfonts,latexsym,graphicx,amsthm}

\theoremstyle{plain}
\newtheorem{theorem}{Theorem}
\newtheorem{lemma}[theorem]{Lemma}
\newtheorem{proposition}[theorem]{Proposition}
\newtheorem{corollary}[theorem]{Corollary}
\newtheorem{remark}[theorem]{Remark}
\theoremstyle{definition}
\newtheorem{definition}[theorem]{Definition}
\newtheorem{example}{Example}
\newtheorem{open}{Open Problem}
\usepackage{newalg}

\makeatletter
\newcommand{\Rmnum}[1]{\expandafter\@slowromancap\romannumeral #1@}
\makeatother

\begin{document}
\title{Approximation for Maximum Surjective Constraint Satisfaction Problems}
\author{Walter Bach\\
CNRS/LIX, \'Ecole Polytechnique, France\\
\texttt{bach@lix.polytechnique.fr}
\and
Hang Zhou \\
D\'epartement d'Informatique, \'Ecole Normale Sup\'erieure, France\\
\texttt{hang.zhou@ens.fr}}
\date{September 21, 2011}
\maketitle

\begin{abstract}
Maximum surjective constraint satisfaction problems (\textsc{Max-Sur-CSP}s) are computational problems where we are given a set of variables denoting values from a finite domain $B$ and a set of constraints on the variables. A solution to such a problem is a \emph{surjective} mapping from the set of variables to $B$ such that the number of satisfied constraints is \emph{maximized}. We study the approximation performance that can be achieved by algorithms for these problems, mainly by investigating their relation with \textsc{Max-CSP}s (which are the corresponding problems without the surjectivity requirement). Our work gives a complexity dichotomy for \textsc{Max-Sur-CSP}$(\mathcal{B})$ between PTAS and APX-complete, under the assumption that there is a complexity dichotomy for \textsc{Max-CSP}$(\mathcal{B})$ between PO and APX-complete, which has already been proved on the Boolean domain and 3-element domains.
\end{abstract}

\section{Introduction}

The constraint satisfaction problem (CSP) is an important computational problem in NP where the task is to decide whether there exists a homomorphism from a given input relational structure $\mathcal{A}$ to a fixed template relational structure $\mathcal{B}$, denoted by CSP($\mathcal{B}$). This problem has many applications in various fields of computer science, such as artificial intelligence and database theory. It is conjectured that CSP($\mathcal{B}$) for any finite structure $\mathcal{B}$ is either in P or NP-complete \cite{Feder1993}. This conjecture has been proved for the Boolean domain \cite{schaefer1978complexity} and 3-element domains \cite{bulatov2006dichotomy}.

The optimization version of CSP($\mathcal{B}$), denoted by \textsc{Max-CSP}($\mathcal{B}$), consists in finding a function from $\mathcal{A}$ to $\mathcal{B}$ satisfying the maximum number of constraints in $\mathcal{A}$ (see, e.g., \cite{creignou2001complexity}). It is conjectured that for any finite $\mathcal{B}$, \textsc{Max-CSP}($\mathcal{B}$) is either in PO or APX-complete \cite{creignou2001complexity}. This conjecture has been proved for the Boolean domain \cite{creignou2001complexity} and 3-element domains \cite{jonsson2006approximability}.

The surjective version of CSP($\mathcal{B}$), denoted by \textsc{Sur-CSP}($\mathcal{B}$), consists in finding a homomorphism using all available values in $\mathcal{B}$ (see, e.g., \cite{survey}). Some problems of the form \textsc{Sur-CSP}($\mathcal{B}$) have been proved to be NP-hard, such as \textsc{Sur-CSP}($C_4^{\text{ref}}$) \cite{paulusma2011computational}; some have unknown computational complexity, such as \textsc{Sur-CSP}($C_6$) \cite{feder1999list} and the No-Rainbow-Coloring problem \cite{bodirsky2004constraint}.

In this article we are interested in the combination of the above two branches of research, that is, the \emph{maximum surjective constraint satisfaction problem} for $\mathcal{B}$, denoted by \textsc{Max-Sur-CSP}($\mathcal{B}$), which consists in finding a surjective function from $\mathcal{A}$ to $\mathcal{B}$ satisfying the maximum number of constraints in $\mathcal{A}$. There are many natural computational problems of the form \textsc{Max-Sur-CSP}($\mathcal{B}$). For example, given a graph with red and blue edges, we might ask for a non-trivial partition of the vertices into two sets such that the sum of the number of red edges between the two sets and the number of blue edges inside the first set is minimized.
This corresponds to the \textsc{Max-Sur-CSP}($\mathcal{B}$) where $\mathcal{B}$ is the relational structure on the Boolean domain with two binary relations: equality and disjunction. We will see this problem later, which is known the \emph{Minimum-Asymmetric-Cut problem}.

It is easy to see (see Proposition~\ref{APX}) that all \textsc{Max-Sur-CSP}s are in the complexity
class APX. Our main result is a complexity dichotomy for \textsc{Max-Sur-CSP}($\mathcal{B}$) on the Boolean domain and 3-element domains: every such problem is either APX-complete or has a \emph{PTAS (Polynomial-Time Approximation Scheme)}. Interestingly, unlike \textsc{Max-CSP}($\mathcal{B}$), there are finite structures $\mathcal{B}$ such that \textsc{Max-Sur-CSP}($\mathcal{B}$) is NP-hard but has a PTAS.

The article is organized as follows. In Section 2, we formally introduce \textsc{Max-Sur-CSP}s. In Section 3, we show that  \textsc{Max-Sur-CSP}($\mathcal{B}$) is in APX for any finite $\mathcal{B}$, by providing a constant-ratio approximation algorithm. In Section 4 and 5, we compare the approximation ratio of \textsc{Max-Sur-CSP}($\mathcal{B}$) with that of \textsc{Max-CSP}($\mathcal{B}$). This comparison leads to the inapproximability or approximability for \textsc{Max-Sur-CSP}($\mathcal{B}$), which depends on the desired approximation ratio. In Section 6, we conclude that if there is a complexity dichotomy for \textsc{Max-CSP}$(\mathcal{B})$ between PO and APX-complete, there is also a complexity dichotomy for \textsc{Max-Sur-CSP}$(\mathcal{B})$ between PTAS and APX-complete. An immediate consequence of this observation is the complexity dichotomy for \textsc{Max-Sur-CSP}$(\mathcal{B})$  on the Boolean domain and 3-element domains. In Section 7, we discuss for which structures $\mathcal{B}$ the \textsc{Max-Sur-CSP}$(\mathcal{B})$ problem is in PO and pose several open problems.

\section{Preliminaries}
Let $\sigma$ be some finite relational signature, which consists of $m$ relations $R_1, ..., R_m$ of arity $k_1, ..., k_m$ respectively. When $m=1$, we write $R$ for $R_1$ and $k$ for $k_1$. Consider only finite $\sigma$-structures $\mathcal{A}=(A,R_1^{\mathcal{A}},\dots,R_m^{\mathcal{A}})$ and $\mathcal{B}=(B,R_1^{\mathcal{B}},\dots,R_m^{\mathcal{B}})$, where $A$ and $B$ are underlying domains, and $R_i^{\mathcal{A}}$ and $R_i^{\mathcal{B}}$ are relations on $A$ and $B$ respectively, $i\in \{1,\dots,m\}$. A \emph{homomorphism} from $\mathcal{A}$ to $\mathcal{B}$ is a function $h: A\rightarrow B$ such that, $(a_1,\dots,a_{k_i})\in R_i^{\mathcal{A}}$ implies $(h(a_1),\dots,h(a_{k_i}))\in R_i^{\mathcal{B}}$, for every $i\in \{1,\dots,m\}$. Define $|\mathcal{A}|$ to be $|A|+|R_1^{\mathcal{A}}|+\dots+|R_m^{\mathcal{A}}|$.

The \emph {constraint satisfaction problem} \textsc{CSP}($\mathcal{B}$) takes as input some finite $\mathcal{A}$ and asks whether there is a homomorphism $h$ from $\mathcal{A}$ to the fixed template $\mathcal{B}$. The \emph {surjective constraint satisfaction problem} \textsc{Sur-CSP}($\mathcal{B}$) is defined similarly, only we insist that the homomorphism $h$ be surjective. The \emph {maximum constraint satisfaction problem} \textsc{Max-CSP}($\mathcal{B}$) asks for a function $h$ from $A$ to $B$, such that the number of constraints in $\mathcal{A}$ preserved by $h$ is maximized. The \emph {maximum surjective constraint satisfaction problem} \textsc{Max-Sur-CSP}($\mathcal{B}$) is defined similarly, only we insist that the function $h$ be surjective. Without loss of generality, we assume that $|A|\geq|B|$ in \textsc{Max-Sur-CSP}($\mathcal{B}$), otherwise there is no surjective function from $A$ to $B$.

\begin{example}
Define the relational structure $\mathcal{B}$ on the Boolean domain with two binary relations: equality and disjunction.

\textsc{Sur-CSP}($\mathcal{B}$) is in P, since it has a simple polynomial-time algorithm as follows. Define a connected component to be the maximal set of variables connected together by the equality relation. If there exists some connected component without any disjunction relation between its variables, then we assign 0 to all its variables and 1 to all other variables. Otherwise there is no surjective assignment.

We call \textsc{Max-Sur-CSP}($\mathcal{B}$) the \emph{Minimum-Asymmetric-Cut problem}. In fact, if there is only the equality relation, this problem
reduces to the \emph{Minimum-Cut problem}. The disjunction relation makes this cut asymmetric.
\end{example}

\begin{example}
Define the relational structure $\mathcal{B}$ on the domain $B=\{0,\dots,5\}$ with the binary relation
$$R^\mathcal{B}=\{(x,y)\,:\,x-y \text{ mod } 6 = \pm1\}.$$

This structure is traditionally denoted by $C_6$ in graph theory. The complexity of the \emph{list homomorphism problem} on this structure has been studied in \cite{feder1999list}. The complexity of \textsc{Sur-CSP}($C_6$) is still open \cite{survey}.
\end{example}

\begin{example}
Define the relational structure $\mathcal{B}$ on the domain $B=\{0,1,2,3\}$ with the binary relation
$$R^\mathcal{B}=\{(x,y)\,:\,x-y \text{ mod } 4 = 1 \text{ or }0\}.$$
This structure is denoted by $C_4^{\text{ref}}$.

\textsc{Sur-CSP}($C_4^{\text{ref}}$) has been proved to be NP-complete \cite{paulusma2011computational}.

\end{example}

\begin{example}
Define the relational structure $\mathcal{B}$ on the domain $B=\{0,1,2\}$ with the ternary relation
$$R^\mathcal{B}=\{0,1,2\}^3\backslash\{(x,y,z)\,:\,x,y,z\text{ distinct}\}.$$

\textsc{Sur-CSP}($\mathcal{B}$), also known as the \emph{No-Rainbow-Coloring problem}, has open complexity \cite{bodirsky2004constraint}. We call \textsc{Max-Sur-CSP}($\mathcal{B}$) the \emph{Minimum-Rainbow-Coloring problem}.
\end{example}

An optimization problem is said to be an \emph{NP optimization problem}  \cite{creignou2001complexity} if instances and solutions can be recognized in polynomial time; solutions are polynomial-bounded in the input size; and the objective function can be computed in polynomial time. In this article, we consider only maximization problems of this form, i.e., \emph{NP maximization problems}.

\begin{definition}
A solution to an instance $\mathcal{I}$ of an NP maximization problem $\Pi$ is \emph{$r$-approximate} if it has value $Val$ satisfying $Val/Opt\geq r$, where $Opt$ is the maximal value for a solution of $\mathcal{I}$.\footnote{There are other conventions to define the ratio of an approximate solution, e.g., $max\{Val/Opt,Opt/Val\}$ in \cite{creignou2001complexity} and $1-Val/Opt$ in \cite{papadimitriou1994computational}. These conventions are equivalent for maximization problems. } A \emph{polynomial-time $r$-approximation algorithm} for an NP maximization problem $\Pi$ is an algorithm which, given an instance  $\mathcal{I}$,  computes an $r$-approximate solution in time polynomial in $|\mathcal{I}|$. We say that $r$ is a \emph{polynomial-time approximation ratio} of an NP maximization problem $\Pi$ if there exists a polynomial-time $r$-approximation algorithm for $\Pi$.
\end{definition}

The following definitions can all be found in \cite{creignou2001complexity}, only with a different convention of the approximation ratio.
\begin{definition}
An NP maximization problem is in the class \emph{PO} if it has an algorithm which computes the optimal solution in polynomial time; and is in the class \emph{APX} if it has a polynomial-time $r$-approximation algorithm, where $r$ is some constant real number in $(0,1]$.
\end{definition}

\begin{definition}
We say that an NP maximization problem $\Pi$ has a \emph{Polynomial-Time Approximation Scheme (PTAS)} if there is an approximation algorithm that takes as input both an instance $\mathcal{I}$ and a fixed rational parameter $\epsilon>0$, and outputs a solution which is $(1-\epsilon)$-approximate in time polynomial in $|\mathcal{I}|$. The class of optimization problems admitting a PTAS algorithm is also denoted by \emph{PTAS}.
\end{definition}

\begin{definition}
An NP maximization problem $\Pi_1$ is said to be \emph{AP-reducible} to an NP maximization problem $\Pi_2$ if there are two polynomial-time computable functions $F$ and $G$ and a constant $\alpha$ such that
\begin{enumerate}
\item for any instance $\mathcal{I}$ of $\Pi_1$, $F(\mathcal{I})$ is an instance of $\Pi_2$;
\item for any instance $\mathcal{I}$ of $\Pi_1$, and any feasible solution $s'$ of $F(\mathcal{I})$, $G(\mathcal{I},s')$ is a feasible solution of $\mathcal{I}$;
\item for any instance $\mathcal{I}$ of $\Pi_1$, and any $r\leq1$, if $s'$ is an $r$-approximate solution of $F(\mathcal{I})$ then $G(\mathcal{I},s')$ is an $(1-(1-r)\alpha-o(1))$-approximate solution of $\mathcal{I}$, where the $o$-notation is with respect to $|\mathcal{I}|$.
\end{enumerate}

An NP maximization problem is \emph{APX-hard} if every problem in APX is AP-reducible to it. It is a well-known fact (see, e.g., \cite{creignou2001complexity}) that AP-reductions compose, and that if $\Pi_1$ is AP-reducible to $\Pi_2$ and $\Pi_2$ is in PTAS (resp., APX), then so is $\Pi_1$.
\end{definition}

\section{\textsc{Max-Sur-CSP}($\mathcal{B}$) is in APX}
In this section, we first provide a simple algorithm to show that every \textsc{Max-CSP}$(\mathcal{B})$ is in APX. Then we adapt this algorithm in order to show that every \textsc{Max-Sur-CSP}$(\mathcal{B})$ is also in APX.
\begin{proposition}
\textsc{Max-CSP}$(\mathcal{B})$ is in APX for any finite $\mathcal{B}$.
\end{proposition}
\begin{proof}
We compute a function $h$ from $A$ to $B$ randomly, i.e., for every $a\in A$, choose $h(a)$ uniformly at random from $B$. Every $k_i$-tuple in $R_i^\mathcal{A}$ is preserved by $h$ with probability $\frac{|R_i^\mathcal{B}|}{|B|^{k_i}}$. Thus we get a randomized $r$-approximation algorithm, where $r=\min_i{\frac{|R_i^\mathcal{B}|}{|B|^{k_i}}}$ is the expected ratio over all random choices. This algorithm can be derandomized via conditional expectations \cite{motwani1995randomized}. In fact, it suffices to select at each step the choice with the largest expected number of satisfied constraints. The expected number of satisfied constraints under some partial assignment can be computed in polynomial time. In this way, we obtain a deterministic $r$-approximation algorithm which runs in polynomial time. Thus \textsc{Max-CSP}($\mathcal{B}$) is in APX.
\end{proof}

Since the function $h$ so obtained is not necessarily surjective, the algorithm presented in the proof does not show that  \textsc{Max-Sur-CSP}($\mathcal{B}$) is in APX. To resolve this problem, the idea is to fix \emph{some} function values of $h$ at the beginning, and to choose the other function values randomly. We start by the simple case where the signature $\sigma$ consists of only one relation. We have the following lemma:
\begin{lemma}
For any $0<r<\frac{|R^\mathcal{B}|}{|B|^k}$, there exists a randomized $r$-approximation algorithm for \textsc{Max-Sur-CSP}$(\mathcal{B})$ which runs in polynomial time.
\end{lemma}

\begin{proof}
Consider the following algorithm, which we call \textsc{Approx}. First, sort the elements in $A$ in increasing order of their \emph{degrees}, which is defined to be the total number of occurrences among all tuples in all relations in $\mathcal{A}$. Then construct an arbitrary bijective function on the set of the first $|B|$ elements in $A$ into the set $B$. Finally, extend this function onto the whole set of $A$ by choosing the other function values uniformly at random. This algorithm runs in polynomial time, and always returns a surjective solution. Let us analyze its approximation performance.

Let $Val$ be the number of preserved $k$-tuples in $R^\mathcal{A}$ in our solution, and $Opt$ be the maximum possible number of preserved $k$-tuples in $R^\mathcal{A}$. Since the sum of degrees of all elements in $A$ is $k|R^\mathcal{A}|$, the sum of the $|B|$ smallest degrees is at most $k|R^\mathcal{A}|\cdot\frac{|B|}{|A|}$. So there are at least ($|R^\mathcal{A}|- k|R^\mathcal{A}|\cdot\frac{|B|}{|A|}$) $k$-tuples in $R^\mathcal{A}$ which are not incident with any of the first $|B|$ elements in $A$. Every such $k$-tuple is satisfied with probability $\frac{|R^\mathcal{B}|}{|B|^k}$ under uniformly random choices. So we have:

\renewcommand{\arraystretch}{3}
$\begin{array}{llll}
\displaystyle\frac{\mathbb{E}[Val]}{Opt} & \geq  & \displaystyle\frac{1}{Opt}\bigg(|R^\mathcal{A}|- k|R^\mathcal{A}|\cdot\frac{|B|}{|A|}\bigg)\cdot\frac{|R^\mathcal{B}|}{|B|^k}\\
& \geq  & \displaystyle \bigg (1- k\cdot\frac{|B|}{|A|}\bigg )\cdot\frac{|R^\mathcal{B}|}{|B|^k}\, & (\text{since } Opt\leq|R^\mathcal{A}|),\\
  \end{array}$\\

\noindent where $k$, $|B|$ and $|R^\mathcal{B}|$ are constant, only $|A|$ is decided by the input instance. Thus for any $0<r<\frac{|R^\mathcal{B}|}{|B|^k}$, we have a randomized polynomial-time $r$-approximation algorithm: follow the \textsc{Approx} algorithm when $|A|>\frac{k|B|}{1-r\frac{|B|^k}{|R^\mathcal{B}|}}$, and tabulate the solution otherwise.
\end{proof}

In general, $\sigma$ consists of $m$ relations $R_1,\dots, R_m$ with arity $k_1,\dots,k_m$ respectively.  A similar analysis leads to:

\renewcommand{\arraystretch}{3}
$\begin{array}{lll}
\displaystyle\frac{\mathbb{E}[Val]}{Opt} & \geq  & \displaystyle\frac{1}{opt}\bigg(\sum_i{|R_i^\mathcal{A}|}- \Big(\sum_i{k_i}|R_i^\mathcal{A}|\Big)\frac{|B|}{|A|}\bigg)\cdot\Big(\min_i{\frac{|R_i^\mathcal{B}|}{|B|^{k_i}}}\Big)\\
& \geq  & \displaystyle \bigg (1- k_{max}\cdot\frac{|B|}{|A|}\bigg )\cdot\Big(\min_i{\frac{|R_i^\mathcal{B}|}{|B|^{k_i}}}\Big)\qquad (\text{since } Opt\leq\sum_i{|R_i^\mathcal{A}|}).\\
  \end{array}$\\
Thus we have a randomized polynomial-time $r$-approximation algorithm, for any $0<r<\min_i{\frac{|R_i^\mathcal{B}|}{|B|^{k_i}}}$.

\begin{proposition}
\label{APX}
\textsc{Max-Sur-CSP}$(\mathcal{B})$ is in APX for any finite $\mathcal{B}$.
\end{proposition}

\begin{proof}
We apply again the derandomization via conditional expectation to get a deterministic polynomial-time $r$-approximation algorithm, for any $0<r<\min_i{\frac{|R_i^\mathcal{B}|}{|B|^{k_i}}}$.
\end{proof}

\section{Inapproximability for \textsc{Max-Sur-CSP}$(\mathcal{B})$}
In this section, we show that \textsc{Max-Sur-CSP}($\mathcal{B}$) is \emph{harder} than \textsc{Max-CSP}($\mathcal{B}$) for the same $\mathcal{B}$, in the sense that the approximation ratio of \textsc{Max-Sur-CSP}($\mathcal{B}$) cannot exceed the approximation ratio of \textsc{Max-CSP}($\mathcal{B}$); and that the APX-hardness of \textsc{Max-CSP}$(\mathcal{B})$ implies the APX-hardness of \textsc{Max-Sur-CSP}$(\mathcal{B})$.

\begin{theorem}
\label{inapproximability}
Let $r\in(0,1]$. If \textsc{Max-CSP}$(\mathcal{B})$ is not polynomial-time $r$-approximable, neither is \textsc{Max-Sur-CSP}$(\mathcal{B})$.
\end{theorem}

\begin{proof}
Suppose there is a polynomial-time $r$-approximation algorithm for \textsc{Max-Sur-CSP}($\mathcal{B}$). We will prove that there is also a polynomial-time $r$-approximation algorithm for \textsc{Max-CSP}($\mathcal{B}$), which causes a contradiction.

Given an instance $\mathcal{A}$ of \textsc{Max-CSP}($\mathcal{B}$), construct an instance $\mathcal{A}'$ of \textsc{Max-Sur-CSP}($\mathcal{B}$) as following: extend the underlying domain $A$ to $A'$ by adding $|B|$ new elements, and keep $R_i^{\mathcal{A}'}$ to be the same as $R_i^{\mathcal{A}}$, for every $i\in \{1,..,m\}$. The optimum of \textsc{Max-Sur-CSP}($\mathcal{B}$) with the instance $\mathcal{A}'$ equals that of \textsc{Max-CSP}($\mathcal{B}$) with the instance $\mathcal{A}$. Any $r$-approximate solution of \textsc{Sur-Max-CSP}($\mathcal{B}$) with the instance $\mathcal{A}'$ is also an $r$-approximate solution of \textsc{Max-CSP}($\mathcal{B}$) with the instance $\mathcal{A}$. Since $\mathcal{A}$ is arbitrary, we thus have a polynomial-time $r$-approximation algorithm for \textsc{Max-CSP}($\mathcal{B}$).
\end{proof}

\begin{proposition}
\label{APX-hardness}
If \textsc{Max-CSP}$(\mathcal{B})$ is APX-hard, \textsc{Max-Sur-CSP}$(\mathcal{B})$ is also APX-hard.
\end{proposition}

\begin{proof}
Given an instance $\mathcal{A}$ of \textsc{Max-CSP}($\mathcal{B}$), construct an instance $\mathcal{A}'$ of \textsc{Max-Sur-CSP}($\mathcal{B}$) as in the above proof. The optimum of \textsc{Max-Sur-CSP}($\mathcal{B}$) with the instance $\mathcal{A}'$ equals that of \textsc{Max-CSP}($\mathcal{B}$) with the instance $\mathcal{A}$. So we have an AP-reduction from \textsc{Max-CSP}($\mathcal{B}$) to \textsc{Max-Sur-CSP}($\mathcal{B}$), where the constant $\alpha$ in the definition of the AP-reduction is 1. Since AP-reductions compose, we then proved the proposition.
\end{proof}

\begin{corollary}
\textsc{Max-Sur-CSP}$(C_6)$ is APX-hard and, under the unique games conjecture\footnote{A formal description of this conjecture could be found in \cite{Khot2002}.}, any polynomial-time approximation ratio of \textsc{Max-Sur-CSP}$(C_6)$ is at most $\alpha_{GW}$(=0.878\dots).
\end{corollary}

\begin{proof}
\textsc{Max-CSP}($C_6$) is exactly the same problem as \textsc{Max-Cut}. From the APX-hardness of \textsc{Max-Cut}, we have the APX-hardness of \textsc{Max-Sur-CSP}$(C_6)$ by Proposition~\ref{APX-hardness}.
The best approximation ratio of \textsc{Max-Cut} has been proved to be $\alpha_{GW}$, under the unique games conjecture \cite{khot2008optimal}. We then deduce from Theorem~\ref{inapproximability} that any polynomial-time approximation ratio of \textsc{Max-Sur-CSP}($C_6$) is at most $\alpha_{GW}$.
\end{proof}

\section{Approximability for \textsc{Max-Sur-CSP}$(\mathcal{B})$}

In this section we describe a PTAS for  \textsc{Max-Sur-CSP}($\mathcal{B}$), given that \textsc{Max-CSP}($\mathcal{B}$) is in PO. The generalized result is stated in the following theorem.

\begin{theorem}
\label{approximability}
Let $r\in(0,1]$. If \textsc{Max-CSP}$(\mathcal{B})$ is polynomial-time $r$-approximable, \textsc{Max-Sur-CSP}$(\mathcal{B})$ is polynomial-time $(r-\epsilon)$-approximable, for any $\epsilon>0$.
\end{theorem}

\begin{proof}

Let \textsc{Approx1} be a polynomial-time $r$-approximation algorithm for \textsc{Max-CSP}($\mathcal{B}$). Consider first the following randomized algorithm:\\

\begin{algorithm}{Approx2}{\mathcal{A}}
h\=\CALL{Approx1}(\mathcal{A})\\
\begin{IF} {h\text{ is surjective}}
\RETURN h
\ELSE
h^*\=h\\
\begin{FOR}{\EACH b \in B \text{ such that $y$ is not an image of $h^*$}}
T\=\{x\in A \;| \;\exists\; y\in A\backslash \{x\}, \;\text{s.t.}\; h^*(x)=h^*(y)\}\\
\begin{IF}
{T=\emptyset}{\RETURN \text{\emph{No Solution.}}}
\end{IF}\\
x\=\text{an element in $T$ chosen uniformly at random}\\
\end{FOR}
h^*(x)\=y\\
\RETURN h^*.
\end{IF}
\end{algorithm}

Let us analyze the performance of this algorithm. Let $Val$ and $Val^*$ be the number of satisfied constraints in $h$ and $h^*$ respectively. Let $Opt$ and $Opt^*$ be the optimum of \textsc{Max-CSP}($\mathcal{B}$) and \textsc{Max-Sur-CSP}($\mathcal{B}$) respectively. Obviously, $Opt^*\leq Opt$.

When $h$ is surjective, we have $h^*=h$ and $\displaystyle\frac{\mathbb{E}[Val^*]}{Opt^*}\geq\frac{\mathbb{E}[Val]}{Opt}.$

When $h$ is not surjective, it is more complicated.  Let $\delta$ be the number of elements in $B$ which are not images of $h$.
For a given tuple $(x_1,\dots,x_k)$ in $A^k$ and another given tuple $(b_1,\dots,b_k)$ in $B^k$, define $p$ to be the probability that $h$ maps $(x_1,\dots,x_k)$ to $(b_1,\dots,b_k)$ and $p^*$ to be the probability that $h^*$ maps $(x_1,\dots,x_k)$ to $(b_1,\dots,b_k)$. Let $T=\{x\in A \;| \;\exists\; y\in A\backslash \{x\}, \;\text{s.t.}\; h(x)=h(y)\}$ and $t=|T|$. Obviously $t>|A|-|B|$. For every $x\in T$, the algorithm above will not modify $h^*(x)$ if there is no other variables in $A$ with the same function value $h(x)$; otherwise it will modify $h^*(x)$ with probability $\frac{\delta}{t}<\frac{|B|}{|A|-|B|}$. Let $Q$ be the conditional probability that $h^*$ does not map $(x_1,\dots,x_k)$ to $(b_1,\dots,b_k)$, given that $h$ maps $(x_1,\dots,x_k)$ to $(b_1,\dots,b_k)$. \\We have:

\renewcommand{\arraystretch}{1.5}
$\begin{array}{llll}
Q& \leq  & \sum_{i=1}^k{\mathbb{P}[\text{The algorithm modifies }h^*(x_i)]}\\
& \leq  & k\cdot\mathbb{P}[\text{ The algorithm modifies $h^*(x)$, for an arbitrary } x\in T]\\
& <  & k\cdot\displaystyle\frac{|B|}{|A|-|B|},\\
  \end{array}$\\

\noindent which implies that  $p^*>\Big(1-k\cdot\displaystyle\frac{|B|}{|A|-|B|}\Big) p$.\\
\noindent The second line of the inequality above holds because the function value of $h^*$ is modified with the same positive probability at every variable in $T$, and with probability 0 at every variable outside $T$. \\

\noindent Thus we have $\displaystyle\frac{\mathbb{E}[Val^*]}{Opt^*}\geq \displaystyle\frac{\mathbb{E}[Val^*]}{Opt}>\Big(1-k_{max}\cdot\frac{|B|}{|A|-|B|}\Big) \frac{\mathbb{E}[Val]}{Opt},$ \\
where $k_{max}$ is the maximal arity of all relations.\\

So whether $h$ is surjective or not, we always have:
$$\displaystyle\frac{\mathbb{E}[Val^*]}{Opt^*}>\Big(1-k_{max}\cdot\frac{|B|}{|A|-|B|}\Big) \frac{\mathbb{E}[Val]}{Opt}.$$
\noindent Hence for any $\epsilon>0$, we can achieve a randomized $(r-\epsilon)$-approximation algorithm for \textsc{Max-Sur-CSP}($\mathcal{B}$): first, precalculate the optimal assignment for every input relational structure whose domain size is less than $N_0=\displaystyle\lfloor\frac{r|B|k_{max}}{\epsilon}\rfloor+|B|$. The running time of this part depends only on $\epsilon$, but not on $|\mathcal{A}|$. Then for a given relational structure $\mathcal{A}$, output a precalculated solution if $|A|\leq N_0$; and execute \textsc{Approx2} with the instance $\mathcal{A}$ otherwise.

The algorithm \textsc{Approx2} can be derandomized by enumerating all possible choices in line 9. This enumerating procedure runs in time polynomial in $|\mathcal{A}|$, since there are polynomially many combinations of choices, and each combination of choices corresponds to a polynomial number of steps in the running time. Thus the derandomized algorithm still runs in time polynomial in $|\mathcal{A}|$. So we get a deterministic  polynomial-time $(r-\epsilon)$-approximation algorithm for \textsc{Max-Sur-CSP}($\mathcal{B}$).
\end{proof}

\begin{proposition}
\label{PTAS}
If \textsc{Max-CSP}$(\mathcal{B})$ is in PO, then \textsc{Max-Sur-CSP}$(\mathcal{B})$ is in PTAS and is not APX-hard, unless $P=NP$.
\end{proposition}

\begin{proof}
\textsc{Max-CSP}($\mathcal{B}$) being in PO implies that it has a polynomial-time approximation algorithm with approximation ratio 1. Theorem~\ref{approximability} then leads to the PTAS-containment. On the other hand, we already know that there are problems in APX but not in PTAS, unless P=NP (see, e.g.,  \cite{hastad1997some}). These problems cannot be AP-reduced to \textsc{Max-Sur-CSP}($\mathcal{B}$), otherwise \textsc{Max-Sur-CSP}($\mathcal{B}$) could not be in PTAS. So \textsc{Max-Sur-CSP}($\mathcal{B}$) is not APX-hard, unless P=NP.
\end{proof}

\begin{corollary}
The Minimum-Asymmetric-Cut problem, \textsc{Max-Sur-CSP}$(C_4^{\text{\upshape{ref}}})$, and the Minimum-Rainbow-Coloring problem are all in PTAS.
\end{corollary}

\begin{proof}
For each of the three problems above, if the function $h$ is not required to be surjective, we can assign 1 to every variable so that all constraints are satisfied. Thus the corresponding \textsc{Max-CSP}($\mathcal{B}$) problem is in PO. We then conclude by applying Proposition~\ref{PTAS}.
\end{proof}

\section {Complexity dichotomy for \textsc{Max-Sur-CSP}($\mathcal{B}$)}
In this section, we first propose a conditional complexity dichotomy for \textsc{Max-Sur-CSP}$(\mathcal{B})$ between PTAS and APX-complete. This condition has already been proved on the Boolean domain and 3-element domains.
\begin{theorem}
\label{dichotomy}
If there is a complexity dichotomy for \textsc{Max-CSP}$(\mathcal{B})$ between PO and APX-complete, there is also a complexity dichotomy for \textsc{Max-Sur-CSP}$(\mathcal{B})$ between PTAS and APX-complete.
\end{theorem}

\begin{proof}
This result is a combination of Proposition~\ref{APX-hardness} and Proposition~\ref{PTAS}.
\end{proof}

\subsection{On the Boolean domain}
\begin{definition}
A constraint $R$ is said to be
\begin{itemize}
\item \emph{0-valid} if $(0,...,0)\in R$.
\item \emph{1-valid} if $(1,...,1)\in R$.
\item \emph{2-monotone} if $R$ is expressible as a DNF-formula either of the form $(x_1\wedge ...\wedge x_p)$ or $(\overline{y_1}\wedge...\wedge \overline{y_q})$ or  $(x_1\wedge ...\wedge x_p) \vee (\overline{y_1}\wedge...\wedge \overline{y_q})$.
\end{itemize}
A relational structure $\mathcal{B}$ is \emph{0-valid (resp. 1-valid, 2-monotone)} if every constraint in $\mathcal{B}$ is 0-valid (resp. 1-valid, 2-monotone).
\end{definition}

Khanna and Sudan have proved that if $\mathcal{B}$ is 0-valid or 1-valid or 2-monotone then \textsc{Max-CSP}$(\mathcal{B})$ is in PO; otherwise it is APX-hard (Theorem 1 in \cite{khanna1996optimization}). In fact, when $\mathcal{B}$ is 0-valid (resp. 1-valid), it suffices to assign 0 (resp. 1) to all variables; and when $\mathcal{B}$ is 2-monotone, we can reduce this problem to \textsc{Min-Cut}, which can be solved efficiently using e.g., the Edmonds-Karp algorithm. Together with Theorem~\ref{dichotomy} , we have the following result.

\begin{theorem}
Let $B=\{0,1\}$. If $\mathcal{B}$ is 0-valid or 1-valid or 2-monotone then \textsc{Max-Sur-CSP}$(\mathcal{B})$ is in PTAS; otherwise it is APX-hard.
\end{theorem}

Using a similar idea as in the proof of Theorem 1 in \cite{khanna1996optimization}, we have the following lemma.
\begin{lemma}
If $\mathcal{B}$ is 2-monotone, then \textsc{Max-Sur-CSP}$(\mathcal{B})$ is in PO.
\end{lemma}
\begin{proof}
We reduce the problem of finding the maximum number of satisfiable constraints to the problem of finding the minimum number of violated constraints. This problem, in turn, can be reduced to the $s-t$ \textsc{Min-Cut} problem in directed graphs. Recall that there are three forms of 2-monotone formulas: (a) $x_1\wedge ...\wedge x_p$, (b) $\overline{y_1}\wedge...\wedge \overline{y_q}$, and (c)$ (x_1\wedge ...\wedge x_p) \vee (\overline{y_1}\wedge...\wedge \overline{y_q})$.

Construct a directed graph $G$ with two special nodes $F$ and $T$ and a different node corresponding to each variable in the input instance. Let $\infty$ denote an integer larger than the total number of constraints. We first select two constraints $C'$ and $C''$, such that $C'$ is of the form (a) or (c), and $C''$ is of the form (b) or (c). Intuitively, $C'$ is to ensure that there exists some element assigned to 1, and $C''$ is to ensure that there exists some element assigned to 0. For a fixed pair $C'$ and $C''$, we proceed as follows for each of the three forms of constraints :
\begin{itemize}
\item For a constraint $C$ of the form (a), create a new node $e_C$ and add an edge from each $x_i$ to $e_C$ of cost $\infty$. Add also an edge from $e_C$ to $T$. This edge is of cost $\infty$ if $C$ is selected as $C'$, and of unit cost otherwise.
\item For a constraint $C$ of the form (b), create a new node $\overline{e_C}$ and add an edge from $\overline{e_C}$ to each $y_i$ of cost $\infty$. Add also an edge from $F$ to $\overline{e_C}$. This edge is of cost $\infty$ if $C$ is selected as $C''$, and of unit cost otherwise.
\item Finally, for a constraint $C$ of the form (c), we create two nodes $e_C$ and $\overline{e_C}$ and connect $e_C$ to each $x_i$ and connect $\overline{e_C}$ to each $y_i$ as described above. If $C$ is selected both as $C'$ and $C''$, then add an edge from $e_C$ to $T$ of cost $\infty$ and an edge from $F$ to $\overline{e_C}$ of cost $\infty$; otherwise if $C$ is only selected as $C'$, then add an edge from $e_C$ to $T$ of cost $\infty$; otherwise if $C$ is only selected as $C''$, then add an edge from $F$ to $\overline{e_C}$ of cost $\infty$; otherwise add an edge from $e_C$ to $\overline{e_C}$ of unit cost.
\end{itemize}
From the correspondence between cuts and assignments, by setting variables on the $T$ side of the cut to be 1 and variables on the $F$ side of the cut to be 0, we find that the cost of a minimum cut separating $T$ from $F$, equals the minimum number of constraints that are violated, under the condition that $C'$ and $C''$ are selected to ensure the surjectivity. We only need to go over all possible combinations of $C'$ and $C''$ to achieve the maximum surjective solution.
\end{proof}

We close this section by providing certain \textsc{Max-Sur-CSP}$(\mathcal{B})$ problems which belong to a complexity class different from that of any \textsc{Max-CSP}$(\mathcal{B})$ problem on the Boolean domain or 3-element domains, under the assumption that P$\neq$NP.

\begin{lemma}
There exist finite relational structures $\mathcal{B}$ such that \textsc{Max-Sur-CSP}$(\mathcal{B})$ is NP-hard but in PTAS.
\end{lemma}

\begin{proof}
Let $\mathcal{B}$ be the relational structure on the Boolean domain with the following relation

$$R^\mathcal{B}=\{(0,0,0),(0,1,1),(1,0,1),(1,1,0),(0,1,0)\}.$$

This relation is not \emph{weakly positive}, not \emph{weakly negative}, not \emph{affine}, and not \emph{bijunctive}\footnote{Their definitions can be found in \cite{CreignouHerbrard}.}, so \textsc{Sur-CSP}($\mathcal{B}$) is NP-hard (which follows from Theorem 4.10 in \cite{CreignouHerbrard}). Since there is a trivial reduction from \textsc{Sur-CSP}($\mathcal{B}$) to \textsc{Max-Sur-CSP}($\mathcal{B}$), we have the NP-hardness of \textsc{Max-Sur-CSP}($\mathcal{B}$). On the other hand, \textsc{Max-CSP}($\mathcal{B}$) is in PO, since we can assign 0 to every variable to satisfy all constraints. So \textsc{Max-Sur-CSP}($\mathcal{B}$) is in PTAS by Proposition~\ref{PTAS}.
\end{proof}

\begin{remark}
In fact, such $\mathcal{B}$ on the Boolean domain is not unique. Any 0-valid or 1-valid or 2-monotone relational structure which is not weakly positive, not weakly negative, not affine, and not bijunctive satisfies the desired property.
There also exist relational structures on larger domains with this property, such as $C_4^{\text{\upshape{ref}}}$ in Example 3.
\end{remark}

\subsection{On 3-element domains}
Jonsson, Klasson, and Krokhin has proved the complexity dichotomy for \textsc{Max-CSP}$(\mathcal{B})$ between PO and APX-hard on 3-element domains (Theorem 3.1 in \cite{jonsson2006approximability}). Together with Theorem~\ref{dichotomy}, we have the following result.

\begin{theorem}
Let $\mathcal{B}$ be any finite relational structure on a 3-element domain. \textsc{Max-Sur-CSP}$(\mathcal{B})$ is either in PTAS or APX-hard.
\end{theorem}

\begin{remark}
The detailed statement of this complexity dichotomy needs concepts of cores and supermodularity, see, e.g., \cite{jonsson2006approximability}.
\end{remark}

\section{Further research}
For \textsc{Max-CSP}$(\mathcal{B})$, define its \emph{approximation threshold} $r_0$ to be the supremum of the set of polynomial-time approximation ratios of \textsc{Max-CSP}$(\mathcal{B})$.
For any $r>r_0$, \textsc{Max-Sur-CSP}($\mathcal{B}$) is not polynomial-time $r$-approximable (by Proposition~\ref{inapproximability}); and for any $r<r_0$, \textsc{Max-Sur-CSP}($\mathcal{B}$) is polynomial-time $r$-approximable (by Proposition~\ref{approximability}). What is the complexity of \textsc{Max-Sur-CSP}($\mathcal{B}$) when the desired approximation ratio is $r_0$?

\begin{open}
Let $r_0$ be the approximation threshold of \textsc{Max-CSP}$(\mathcal{B})$. Is \textsc{Max-Sur-CSP}$(\mathcal{B})$ polynomial-time $r_0$-approximable?
\end{open}

Consider a special case of the above problem, where $r_0$ is 1, which means that \textsc{Max-CSP}$(\mathcal{B})$ is in PO.

\begin{open}
Given that \textsc{Max-CSP}$(\mathcal{B})$ is in PO. Is \textsc{Max-Sur-CSP}$(\mathcal{B})$ in PO?
\end{open}

In fact, we failed to give an answer, even for some concrete finite structures $\mathcal{B}$ on small domains:

\begin{open}
Are the Minimum-Asymmetric-Cut problem and the Minimum-Rainbow-Coloring problem in PO?
\end{open}

The difficulty of the above problems lies in the classification of the \textsc{Max-Sur-CSP}$(\mathcal{B})$ problems which are in PTAS.
\begin{open}
If P$\neq$NP, is there a complexity trichotomy for \textsc{Max-Sur-CSP}$(\mathcal{B})$ among PO, PTAS but NP-hard, and APX-hard, for any finite structure $\mathcal{B}$ on the Boolean domain or 3-element domains?
\end{open}

\bibliographystyle{plain}
\bibliography{references}

\end{document}